\documentclass{llncs}
\usepackage{comment}
\usepackage{graphicx}
\usepackage{indentfirst}
\usepackage{amsfonts}
\usepackage{verbatim}
\usepackage{amsmath}
\usepackage{listings}
\usepackage[english]{babel}
\usepackage{xcolor}
\usepackage{latexsym,amsmath,epsfig,amssymb,graphics}
\usepackage{caption,multicol}
\usepackage{lipsum}
\usepackage{times}

\newenvironment{Figure}
  {\par\medskip\noindent\minipage{\linewidth}}
    {\endminipage\par\medskip}
	
\usepackage[ruled,vlined,linesnumbered]{algorithm2e}
\newcommand{\oomit}[1]{}

	\newcommand{\keywords}[1]{\par\addvspace\baselineskip
\noindent\keywordname\enspace\ignorespaces#1}

\newcommand{\yy}{\mathbf{y}}
\newcommand{\xx}{\mathbf{x}}

\newcommand{\bfa}{\mathbf{a}}

\newcommand{\RR}{\mathbb{R}}

\def \NN  { {\mathbb N } }

\def \T {{{\mathcal T}}}
\def \I {I }
\def \II { {\mathcal I} }

\def \alg { {\tt SN\_Interpolants} }
\def \ualg { {\tt RSN\_Interpolants}}

\def \sdp {{\bf SDP}}
\def \imply { \Rightarrow}
\def \MM { \mathcal{M}}
\def \CC {\mathcal{C}}
\def \cand {{\tt Certificate\_Generation}}
\def \aisat {{\tt AiSat}}
\def \csdp  {{\tt CSDP}}

\def \sostool {{\tt SOSTOOLS}}

\begin{document}

\title{ Generating Non-Linear Interpolants by Semidefinite Programming}

\author{Liyun Dai\inst{1,3}\thanks{Corresponding author} \and Bican Xia\inst{1} \and Naijun Zhan\inst{2}}
\institute{LMAM \& School of Mathematical Sciences, Peking University \and
State Key Laboratory of Computer Science, Institute of Software, CAS \and
Beijing International Center for Mathematical Research, Peking University \\
 \email{dailiyun@pku.edu.cn ~ xbc@math.pku.edu.cn ~ znj@ios.ac.cn}
}
\date{}
\maketitle
\begin{abstract}
   Interpolation-based techniques have been widely and successfully applied in the verification  of hardware and software, e.g., in bounded-model checking, CEGAR, SMT, etc., whose hardest part is how to synthesize interpolants.
   Various work for discovering interpolants for propositional logic, quantifier-free fragments of first-order theories and their combinations  have been proposed. However, little work focuses on discovering polynomial interpolants in the literature. In this paper, we provide an  approach for constructing non-linear interpolants  based on semidefinite programming, and show how to apply such results to the verification of programs by examples.
\keywords {Craig interpolant, Positivstellensatz Theorem, semidefinite programming, program verification.}
\end{abstract}

\section{Introduction}\label{sec:int}
It becomes a grand challenge how to guarantee the correctness of software, as our modern life more and more
depends on computerized systems. There are lots of verification techniques based on either
   model-checking \cite{CE81}, or  theorem proving \cite{Issabelle,PVS}, or abstract interpretation \cite{cousot77}, or their combination, that have been invented for
 the verification of hardware and software, like bounded model-checking \cite{bmc}, CEGAR \cite{cegar}, satisfiability modulo theories (SMT) \cite{smt}, etc.
  The bottleneck of these techniques is scalability, as many of real software are very complex with different features like complicated data structures, concurrency, distributed, real-time and hybrid, and so on.
    While interpolation-based techniques provide a powerful mechanism for local and modular reasoning, which  indeed improves  the scalability of these techniques, in which the notion of Craig interpolants plays a key role.

   Interpolation-based local and modular reasoning was first applied in theorem proving due to Nelson and Oppen \cite{NelsonOppen}, called Nelson-Oppen method. The basic idea of Nelson-Oppen method is to reduce the satisfiability (validity) of a composite theory into the ones of its component theories whose satisfiability (validity) have been obtained. The hardest part of the method, which also determines the efficiency of the method, is to construct a formula using the common part of the component theories for a given formula of the composite theory with Craig Interpolation Theorem. In the past decade, the Nelson-Oppen method was further
   extended to SMT which is based on DPLL \cite{dpll} and
     Craig Interpolation Theorem \cite{craig} for combining different decision procedures in order to verify a property of programs with complicated data structures. For instance, Z3 \cite{Z3} integrates more than 10 different decision procedures up to now, including propositional logic, equality logic with uninterpreted functions, Presburger arithmetic, array logic, difference arithmetic, bit vector logic, and so on.

  In recent years, it is noted that
    interpolation based local and modular reasoning is quite helpful for improving the scalability of
       model-checking, in particular for bounded model-checking of systems with finite or infinite states \cite{bmc,Mcmilian03,Mcmilian05}, CEGAR \cite{predab,Henzinger04}, etc.
     McMillian first considered how to combine Craig interpolants with bounded model-checking to verify infinite state systems.   The basic idea of his approach is to generate invariant using Craig interpolants so that
     it can be claimed that an infinite state system satisfies a property after $k$ steps in model-checking whenever
     an invariant is obtained which is strong enough to guarantee the property. While in \cite{predab,Henzinger04,Mcmilian06}, how to apply the local property of Craig interpolants generated from
      a counter-example to refine the abstract model in order to exclude the spurious counter-example in CEGAR was investigated.
      Meanwhile, in \cite{wang11}, using interpolation technique to generate a set of atomic predicates as the base of machine-learning based verification technique was investigated by Wang et al.

      Obviously, synthesizing Craig interpolants is the cornerstone of interpolation based techniques. In fact, many approaches have been proposed in the literature.  In \cite{Mcmilian05}, McMillian presented a method for deriving Craig interpolants from proofs in the quantifier-free theory of
 linear inequality and uninterpreted function symbols, and based on which an interpolating theorem prover was
 provided.   For improving the efficiency of constructing interpolant, McMillian further
 proposed a method based on  lazy abstraction for generating interpolants.  While, in  \cite{Henzinger04},   Henzinger et al. proposed a method to synthesizing Craig interpolants for a theory with arithmetic and pointer expressions, as well as call-by-value functions. In \cite{Yorsh05}, Yorsh and Musuvathi  presented a combination method for generating Craig interpolants for a class of first-order
theories. While Rybalchenko and Sofronie-Stokkermans \cite{Rybalcheno10} proposed an approach
by reducing the synthesis of Craig interpolants of  the combined theory of linear arithmetic and
uninterpreted function symbols to constraint solving.

However, in the literature, there is little work on how to synthesize non-linear interpolants, except
 that in \cite{Kupferschmid11}  Kupferschmid and Becker provided a method to construct
non-linear Craig Interpolant using iSAT, which is a variant of SMT solver based on interval arithmetic.

In this paper we investigate how to construct non-linear interpolants. The idea of our approach is as follows:
 Firstly, we reduce the problem of generating interpolants for arbitrary two polynomial formulas to that of
 generating interpolants for two semi-algebraic systems (SASs), which is a conjunction of a set of polynomial equations, inequations and inequalities (see the definition later). Then, according to \textbf{Positivstellensatz Theorem} of real algebraic geometry \cite{BCR98}, there exists a witness to indicate the considered two SASs do not have common
 real solutions if their conjunction is unsatisfiable. Parrilo in \cite{Parrilo00,Parrilo03} gave an approach for constructing the witness by applying semidefinite programming \cite{VB96}.   Our algorithm
invokes Parrilo's method as a subroutine. Our purpose is to construct Craig interpolants, so we need to
obtain a special witness. In general case, we cannot guarantee the existence of the special witness,
   which means that our approach is only sound, but not complete.
    However, we discuss that if the considered two SASs meet
 \emph{Archimedean condition}, (e.g., each variable occurring in the SASs is bounded, which is a reasonable assumption in practice), our approach is not only sound, but also complete.   We demonstrate our approach by
 some examples, in particular, we show how to apply the results to program verification by examples.

The complexity of our approach is polynomial in $u b \binom{n+b/2}{n}\binom{n+b}{n}$, where $u$ is the number of polynomial
constraints in the considered problem, $n$ is the number of variables, and $b$ is the highest degree of
polynomials and interpolants. So, the complexity of our approach is polynomial in $b$ for a given problem as in which $n$ and $u$ are fixed.

\subsubsection{Structure of the paper:} The rest of the paper is organized as follows. By a running example, we sketch our approach and show how to apply it to program verification in Section \ref{sec:over}.
Some necessary preliminaries are introduced in Section \ref{sec:found}. A
sound but incomplete algorithm  for synthesizing  non-linear interpolants in general case is described in Section \ref{sec:sound}.  Section \ref{sec:comp} provides a practical algorithm
 for systems only containing non-strict inequalities and satisfying Archimedean condition.
 Section \ref{sec:cor} focuses on the correctness and complexity analysis of our approach.
Our implementation and experimental results are briefly reported in Section \ref{sec:exp}. 
Section  \ref{sec:rel} describes more related work related to interpolant generation and
its application to program verification. Our summarizes the paper and discusses future work in Section \ref{sec:con}  .

\section {An Overview of Our Approach} \label{sec:over}
In this section, we sketch our approach  and show how to apply our results to program verification by an example.

\begin{center}
\begin{minipage}[b]{0.4\textwidth}
  \begin{lstlisting}[caption=example,label=example]
if (x*x+y*y<1)
  /*  initial values
while (x*x+y*y<3){
 x:=x*x+y-1;
 y:=y+x*y+1;
 if(x*x-2*y*y-4>0)
 /* unsafe area
   error();
 }
  \end{lstlisting}
     \hspace*{3cm}
\end{minipage}%
\begin{minipage}[b]{0.4\textwidth}
  $g_1=1-x^2-y^2>0$ \\
   ~ \\
  $g_2=3-x^2-y^2>0$\\
  $f_1=x^2+y-1-x'=0$\\
  $f_2=y+x'y+1-y'=0 $\\
  $g_3=x'^2-2y'^2-4>0$\\
   ~ \\
   ~\\[2mm]
   ~ \\
  ~ \\
 \end{minipage}\\
 \end{center}
Consider the program in Code~\ref{example} (left).
   This program tests the initial value of  $x$ and $y$ at line 1, afterwards executes the
   while loop with $x^2+y^2<3$ as the loop condition.
    The body of the while loop contains two assignments and
     an \textbf{if} statement in sequence. The property
  we wish to check is that  {\bf error()} procedure will never be executed.
   Suppose there is  an execution $1\rightarrow  3 \rightarrow 4 \rightarrow  5 \rightarrow 6 \rightarrow 8$.
   We can encode such an execution by the formulas as in Code~\ref{example} (right). Note that in these formulas
we use unprimed and primed versions of each variable to represent the values of the variable before and after
   updating respectively.
 Obviously,  the execution is infeasible iff the conjunction of these formulas  is
unsatisfiable. Let $\phi\triangleq g_1>0\wedge f_1=0 \wedge f_2=0\footnotemark$  and $\psi\triangleq g_3>0$. \footnotetext{As $g_1>0 \imply g_2>0$, we ignore $g_2>0$ in $\phi$.}
To show $\phi \wedge \psi$ is unsatisfiable, we need to construct an {\em interpolant} $\theta$ for $\phi$ and $\psi$, i.e., $\phi \Rightarrow \theta$ and $\theta \Rightarrow \neg \psi$. If there exist $\delta_1,\delta_2,\delta_3, h_1,h_2$ such that $$g_1\delta_1+f_1h_1+f_2h_2+g_3\delta_2+\delta_3= -1,$$
  where $\delta_1,
  \delta_2, \delta_3 \in \RR[x, y, x', y']$ are sums of squares and $h_1, h_2\in \RR[x,y,x',y']$,
 then $\theta\triangleq g_3\delta_2+\frac{1}{2} \le  0$ is such an interpolant for $\phi$ and $\psi$.
 In this example, applying our tool \aisat, we obtain in 0.025 seconds that
{\scriptsize  \begin{eqnarray*}
  h_1&=& -290.17 -56.86y' +1109.95x' +37.59y -32.20yy' +386.77yx'  +203.88y^2 +107.91x^2,\\
  h_2&=& -65.71 +0.39y' +244.14x' +274.80y +69.33yy' -193.42yx' -88.18y^2
      -105.63x^2, \\
  \delta_1&=& 797.74 -31.38y' +466.12y'^2 +506.26x' +79.87x'y' +402.44x'^2 +104.43y  \\
     & & +41.09yy'  -70.14yx' +451.64y^2 +578.94x^2 \\
  \delta_2&=&436.45,\\
  \delta_3&=& 722.62 -91.59y' +407.17y'^2 +69.39x' +107.41x'y' +271.06x'^2 +14.23y +188.65yy' \\
          & &   +69.33yy'^2 -600.47yx'
		     -226.01yx'y' +142.62yx'^2 +325.78y^2 -156.69y^2y' +466.12y^2y'^2  \\
		  & &   +10.54y^2x'y' +595.87y^2x'^2 -11.26y^3 +41.09y^3y' +18.04y^3x' +451.64y^4 +722.52x^2 \\
		  & & -80.15x^2y'  +466.12x^2y'^2 -495.78x^2x' +79.87x^2x'y' +402.44x^2x'^2 +64.57x^2y \\
		  & & +241.99y^2x' +73.29x^2yy' -351.27x^2yx' +826.70x^2y^2 +471.03x^4.
  \end{eqnarray*} }
Note that $\delta_1$ can be represented as {\small $923.42(0.90+0.7y-0.1y'+0.43x')^2+252.84(0.42-0.28y+0.21y'-0.84x'
)^2+461.69(-0.1-0.83y+0.44y'+0.34x')^2+478(-0.06+0.48y+0.87y'+0.03x')^2
+578.94(x)^2$.} Similarly,  $\delta_2$ and $\delta_3$ can be represented as sums of squares also.

Moreover, using the approach in \cite{CXYZ07}, we can prove $\theta$ is an inductive invariant of the loop, therefore, {\bf error()} will never be executed.
\begin{remark}
Note that $\theta$ itself cannot be generated using
 quantifier elimination (QE for short) approach in \cite{CXYZ07}, as it contains more than thirty monomials, which means that there are more than thirty parameters at least in any predefined template which can be used to generate $\theta$. Handling so many parameters is far beyond the capability of all the existing tool based on QE.
 However, the problem whether $\theta$ is an inductive variant only contains $4$ variables, therefore it
 can be verified using QE. The detailed comparison between our approach reported in this paper and QE based technique can be seen in the related work.
 \end{remark}

\section{Theoretical Foundations} \label{sec:found}
In this section, for self-containedness, we briefly introduce some basic notions and mathematical theories,
  based on which our approach is developed.

\begin{definition}[Interpolants]\label{def:int}
 A theory  $\T$ has interpolant if for all formulae $\phi$ and $\psi $ in the signature of $\T$, if $\phi \models_{\T} \psi $,  then there exists a formula $\Theta$ that contains only symbols that $\phi$ and $\psi$ share such that $\phi \models_{\T} \Theta$ and $\Theta \models_{\T} \psi$.

 An interpolant $\Theta$ of $\phi$ and $\neg \psi$ is called \emph{inverse interpolant} of $\phi$ and $\psi$, i.e., $\phi \wedge \psi  \models_{\T} \perp$, $\phi \models_{\T} \Theta$ and $\Theta\wedge \psi \models_{\T} \perp$, where $\Theta$ contains only the symbols that $\phi$ and $\psi$ share.
  \end{definition}

Note that in practice, people like to abuse \emph{inverse interpolant} as \emph{interpolant}. Thus, as a convention, in the sequel, all interpolants are referred to \emph{inverse interpolant} if not otherwise stated.

Also, in what follows, we denote by $\xx$ a variable vector $(x_1,\cdots,x_n)$ in $\RR^n$, and by $\RR[\xx]$
  the polynomial ring with real coefficients in variables $\xx$.

 \subsection{Problem Description}
Here, we describe the problem we consider in this paper.
Let
		\begin{eqnarray} \label{eq:subpro}
				   \T_{1t}~=~
				   \bigwedge_{j=0}^{k_t}  f_{tj}(\xx)	\triangleright 0 & \hspace*{.5cm} \mbox{ and } \hspace*{.5cm} &
				   \T_{2l}~=~
				   \bigwedge_{j=0}^{s_l}  g_{lj}(\xx)	\triangleright' 0,
		 \end{eqnarray}
be two semi-algebraic systems (SASs), where  $f_{ij}$ and $g_{ij}$ are polynomials in $\RR[\xx]$, and $\triangleright_{ij}, \triangleright_{ij}'  \in \left\{ =,\neq,\ge \right\}$.
Clearly, any polynomial formula $\phi$ can be represented as a DNF, i.e. the disjunction of a several
SASs. Let
$\T_1=\bigvee_{t=1}^m  \T_{1t}, \T_2 = \bigvee_{l=1}^n  \T_{2l}$ be two polynomial formulas  and $\T_1\wedge\T_2\models \bot$, i.e., $\T_1$ and $\T_2$ do not share any real solutions. Then, the problem to be considered in this paper is how to find another polynomial formula $I$ such that $\T_1\models I$ and $ I \wedge \T_2 \models \bot$.

 It is easy to show that  if, for each $t$ and $l$, there is an interpolant $I_{tl}$ 
for $\T_{1t}$ and $\T_{2l}$,
then $I=\bigvee_{t=1}^m\bigwedge_{l=1}^n I_{tl}$  is an interpolant of $\T_1$ and $\T_2$.
Thus, we only need to consider how to construct interpolants for two SASs of the form  (\ref{eq:subpro}) in the rest of this paper.

\subsection{Common variables}

	 In the above problem description, we assume $\T_1$ and $\T_2$ share a set of variables. But in practice, it is possible that they have different variables. Suppose $\mathcal{V}(\T_i)$ for the set of variables that indeed occur in $\T_i$, for $i=1,2$.
For each $v\in \mathcal{V}(\T_1)-\mathcal{V}(\T_2)$, if $v$ is a local variable introduced
 in the respective program, we always have an equation $v=h$ corresponding to the assignment to $v$ (possibly the composition of a sequence of assignments to $v$); otherwise, $v$ is a global variable, but
 only occurring in $\T_1$, for this case, we introduce an equation $v=v$ to $\T_2$; Symmetrically,
  each $v\in \mathcal{V}(\T_2)-\mathcal{V}(\T_1)$ can be coped with similarly.

  In the following, we show how to derive the equation $v=h$ from the given programs by case analysis.
  \begin{itemize}
  \item  If the given program has no recursion nor loops, we can find out the dependency between
  the variables in $\mathcal{V}({\T}_1)\cap \mathcal{V}(\T_2)$ and the variables in $\mathcal{V}(\T_j)-\mathcal{V}(\T_{3-j})$ according to the order of assignments
  in the program segment, where $j=1,2$. Clearly, we can always represent each variable in
    $\mathcal{V}(\T_j)-\mathcal{V}(\T_{3-j})$ by an expression of $\mathcal{V}({\T}_1)\cap \mathcal{V}(\T_2)$.
    Obviously, if all expressions in the program segment are polynomial, the resulted expressions are polynomial either.

   \item If the given program contains loops or recursion, it will become more complicated. So,
   we have to unwind the loop and  represent each variable in
    $\mathcal{V}(\T_j)-\mathcal{V}(\T_{3-j})$ by an expression of $\mathcal{V}({\T}_1)\cap \mathcal{V}(\T_2)$ and the number $i$ of the iterations of the loops or recursions. However, the resulted expressions may not be polynomial any more. But as proved in \cite{rck07}, if assignment mappings of the loops in the program segment are solvable, the resulted expressions are still polynomial.
    \end{itemize}

 \begin{definition}[Solvable mapping \cite{rck07}]
   Let $g\in \Bbb{Q}[\xx]^m$ be a polynomial mapping. $g$ is solvable if there exists a partition of
 $\xx$ into subvectors of variables, $\xx = \textbf{w}_1\cup\cdots \cup \textbf{w}_k$,
   $\textbf{w}_i \cap \textbf{w}_j =\emptyset$ if $i \neq j$, such that $\forall j : 1 \leq j \leq k$
we have
  \[g_{\textbf{w}_j}(\xx) = M_j\textbf{w}_j^T + \textbf{P}_j(\textbf{w}_1,\ldots,\textbf{w}_{j-1}),\]
where $M_j\in \Bbb{Q}^{|\textbf{w}_j|\times |\textbf{w}_j|}$ is a matrix and $\textbf{P}_j$
  is a vector of $|\textbf{w}_j|$ polynomials in the ring
 $\Bbb{Q}[\textbf{w}_1,\ldots, \textbf{w}_{j-1}]$. For $j = 1$, $\textbf{P}_1$ must be a constant vector, implying that $g_{\textbf{w}_1}$ is an affine
mapping.
\end{definition}

\begin{tabular}[!htb]{ p{0.45\textwidth} p{0.55\textwidth} }
  \begin{lstlisting} [label=sub3,caption=\ \ \ \ demo example]
assume(a+b>0);
 int x=a;int y=a;
 while(x+y<20){
    x=x*a;
    y=x+y*b;
 }
  \end{lstlisting}

  &
  \begin{lstlisting}
 int z=0;int w=0;
 while(w<z+50){
    a=a+1;b=b/2;
    w=w+a; z=w*b;
 }
 assert(x+y+z+w>0);
  \end{lstlisting}
  \tabularnewline
\end{tabular}

 For example, in the Code \ref{sub3} $a,b$ are {\it common variables}, $x,y,z,w$ are {\it local variables}.
Let $\T_1$ be related to the left  and $\T_2$ to the right of Code \ref{sub3}.  $\T_1$ uses an order $y\succ x \succ b \succ a$ on variables, and  $\T_2$  uses an order
  $z \succ w \succ b \succ a$ on variables. Obviously, in every iteration of the {\it  loop}s, variable with higher precedence can only be assigned with a polynomial of variables with lower precedence.
   In order to prove the {\it assert},  we unwind the first loop $i$ times, and
   obtain the values of $a,b,x,y$ are $a,b, a^{i+1}, \sum_{j=0}^{j=i+1}a^{i+1-j}b^j$, respectively.
   Similarly, unwind the second loop $j$ times, and obtain
   the values of $a,b,z,w$ are $a+j, \frac{b}{2^j}, ja+\frac{(j+1)j}{2}, (ja+\frac{(j+1)j}{2})\frac{b}{2^j}$, respectively. Thus, in the first loop, the local variables $x,y$ are represented
   by expressions of $a,b$, so are $z,w$ in the second loop.
   Using such replacements, we can obtain an interpolant $I_{ij}$ only concerning the common variables
    $a,b$ w.r.t. the $i$-th unwinding of the first loop and the $j$-th unwinding of the second loop. Whenever we can prove  that $I_{ij}$ is an invariant of Code \ref{sub3}, then the assert is guaranteed.
    This is a procedure  of {\it BMC}. 


In what follows, we  use
  {\tt subvariable} to denote the above procedure to transform two SASs that may not share same variables to two SASs that share same variables.

\oomit{acome from a given  abstract error trace (which cannot be concretized) and they denote  past and the future segment of the trace,  respectively. We restrict ourselves to the case
	 where every   variable $v$ which only occurs in $\T_2$  there is  a
 equity  $v=h_v$ in $\T_2$ and  every variable  in  $h_v$  occurs earlier in trace than $v$. There is a topical order such  that we can use Algorithm   substitute every uncommon $v$ variable of $I$ by a topology  order with a polynomial $h_v$ obtain $I'$.    According Definition \ref{def:int}, $I'$ is an interpolant of $\T_1, \T_2$. This mean that we will not be able to find
 interpolants when $\T_1,\T_2$ can not satisfy above  restriction. This restriction is equivalent to {\it solvable mapping} in \cite{solve05}. }

\subsection{Real Algebraic Geometry}
In this subsection, we introduce some basic notions and results on real algebraic geometry, that will be used later. 

\begin{definition}[ideal] \label{def:ideal}
				  Let $\II$ be an ideal in $\RR[\xx]$, that is, $\II$ is an additive subgroup of $\RR[\xx]$ satisfying $fg\in \II$ whenever $f\in \II$ and $g\in \RR[\xx]$.
				  Given $h_1,\ldots,h_m\in \RR[\xx]$,
				  					$\left< h_1,\ldots,h_m\right> =  \left\{ \sum_{j=1}^m u_jh_j\ | \ u_1,\ldots,u_m\in \RR[\xx] \right\} $
				  denotes the ideal generated by $h_1,\ldots,h_m$.
				\end{definition}

\begin{definition}[multiplicative monoid]
Given a polynomial set $P$, let $\textit{Mult}(P)$ be the \emph{multiplicative monoid} generated by $P$, i.e., the set of finite products of the elements of $P$ (including the empty product which is defined to be $1$).
 \end{definition}
				
\begin{definition}[Cone]
				   A \emph{cone} $\mathcal{C}$ of $\RR[\xx]$ is a subset of $\RR[\xx]$ satisfying the following conditions: (\textbf{i}) $p_1,p_2\in  \mathcal{C} \Rightarrow p_1+p_2 \in  \mathcal{C}$;
             (\textbf{ii})  $p_1,p_2\in  \mathcal{C} \Rightarrow p_1p_2\in  \mathcal{C}$;
			 (\textbf{iii}) $p \in \RR[\xx]\Rightarrow p^2 \in  \mathcal{C}$.
				 \end{definition}

  Given a set $P\subseteq \RR[\xx],$  let $ \mathcal{C}(P)$ be the smallest cone of $\RR[\xx]$ that contains $P$. It is easy to see that $ \mathcal{C}(\emptyset)$ corresponds to the polynomials that can be represented as a sum of squares, and is the smallest cone in $\RR[\xx]$, i.e., $\left\{\ \sum_{i=1}^sp_i^2 \mid p_1,\ldots, p_s \in \RR[\xx] \right\}$, denoted by $\mathbf{SOS}$.
For a finite set $P\subseteq \RR[\xx]$, $\mathcal{C}(P)$ can be represented as:
				   $$  \mathcal{C}(P)=\{q+\sum_{i=1}^r q_ip_i\ |\ q,q_1,\ldots,q_r\in \mathcal{C}(\emptyset), p_1,\ldots,p_r\in \textit{Mult}(P)\}.$$

 \emph{Positivstellensatz Theorem}, due to Stengle \cite{BCR98}, is an important theorem in real algebraic geometry. It states that, for a given SAS, either
				 the system has a solution in $\RR^n$, or there exists a certain polynomial identity which bears witness to indicate that the system has no solutions.

\begin{theorem}[Positivestellensatz Theorem, \cite{BCR98}] \label{the:2}
Let $(f_j)_{j=1}^{s}, \ (g_k)_{k=1}^{t}, (h_l)_{l=1}^u$ be finite families of polynomials in $\RR[\xx]$. Denote by $\mathcal{C}$ the cone
				   generated by $(f_j)_{j=1}^s$, $\textit{Mult}$ the multiplicative monoid generated by $ (g_k)_{k=1}^t,$ and $\II$ the ideal generated by $(h_l)_{l=1}^u$. Then the following two statements are equivalent:
		   \begin{enumerate}
					 \item the SAS
						$ \left\{
						   \begin{array}{c}
							 \  f_1(\xx)\geq 0,\quad \cdots, \quad f_s(\xx)\geq 0, \\
							 \ g_1(\xx)\neq 0, \quad \cdots, \quad g_t(\xx)\neq 0,  \\
							 \ h_1(\xx)=0, \quad \cdots, \quad  h_u(\xx)=0\\
						   \end{array} \right.$
					   has no real solutions;
					 \item   there exist $f\in  \mathcal{C}$, $g\in \textit{Mult}$, $h \in \II $ such that $f+g^2+h\equiv 0$.
 \end{enumerate}
\end{theorem}

 \subsection{Semidefinite Programming}

 In \cite{BCR98}, Stengle did not provide a constructive proof to Theorem \ref{the:2}.
   However, Parrilo in \cite{Parrilo00,Parrilo03} provided a constructive way to obtain the witness,
    which is based on semidefinite programming. Parrilo's result will be the starting point of our
     method, so we briefly review  semidefinite programming below.
 We use  $Sym_n$ to denote the set of  $n\times n$ real symmetric matrices, and $deg(f)$  the highest total degree  of $f$ for a given polynomial $f$ in the sequel.

 \begin{definition}[Positive semidefinite matrices] A matrix $M\in Sym_n$ is called \emph{positive semidefinite}, denoted by  $M\succeq 0$, if
   $\xx^TM\xx\ge 0$  for all $\xx\in \RR^n$.
 \end{definition}

\begin{definition}[Inner product]\label{def:inner}
The {\em inner product} of two matrices $A=(a_{ij}),B=(b_{ij})\in \RR^{n\times n}$, denoted by
  $\left<A,B\right>$, is defined by
 $Tr(A^TB)=\sum _{i,j=1}^na_{ij}b_{ij}$.
\end{definition}

\begin{definition}[Semidefinite programming (SDP)]\label{def:sdp}
The standard (primal) and dual forms of a {\sdp} are respectively given in the following:
 \begin{eqnarray}
	p^* &=& \inf_{X\in Sym_n}\left<C,X\right> \mbox{ s.t. } X\succeq 0,\ \left<A_j,X\right>=b_j\  (j=1,\ldots,m)
	\label{eq:primal} \\
d^* &= & \sup_{ y\in \RR^m} \mathbf{b}^T\yy \  \mbox{ s.t.} \ \sum_{j=1}^m y_jA_j +S=C,\  S \succeq 0,
	\label{eq:dual}
  \end{eqnarray}
   where  $C,A_1,\ldots,A_m,S\in \textit{Sym}_n$ and $\mathbf{b}\in \RR^m$.
  \end{definition}

There are many  efficient algorithms to solve
 {\sdp} such as interior-point method. We present a basic path-following algorithm to solve (\ref{eq:primal}) in the following.

\begin{definition}[Interior point for \sdp]
  \begin{eqnarray*}
	\textit{intF}_p & = & \left\{ X: \left<A_i,X\right>=b_i\ (i=1,\ldots,m),\ X\succ 0 \right\}, \\
	\textit{intF}_d & = & \left\{ (\mathbf{y},S): S=C-\sum_{i=1}^mA_iy_i\succ 0 \right\}, \\
   \textit{intF} & = & \textit{intF}_p\times \textit{intF}_d.
   \end{eqnarray*}
\end{definition}

Obviously, $\left<C,X \right> -\mathbf{b}^T\yy=\left<X,S\right> >0 $ for all $(X,\yy,S)\in \textit{intF} $. Especially, we have $ d^*\le p^*$. So the soul of interior-point method to compute $p^*$
  is to reduce $\left<X,S\right>$ incessantly and meanwhile guarantee $(X,\yy,S)\in \textit{intF}$.
  \begin{algorithm}[!htb]
  \SetKwData{Left}{left}\SetKwData{This}{this}\SetKwData{Up}{up}
  \SetKwFunction{Union}{Union}\SetKwFunction{FindCompress}{FindCompress}
  \SetKwInOut{Input}{input}\SetKwInOut{Output}{output}
  \Input{ $C$, $A_j, b_j\ (j=1,\dots,m)$ as in (\ref{eq:primal}) and a threshold $c$ }
  \Output{ $p^* $ }
  \SetAlgoLined
  \BlankLine
  Given a $(X,\yy,S)\in \textit{intF}$ and $XS=\mu I$\;
  \tcc{ $\mu$ is a positive constant and $I$ is the identity matrix. }
  \While{$\mu>c$  }{
	$\mu=\gamma\mu$\;
	\tcc{$\gamma$ is a fixed positive constant less than one}
	use  {\bf Newton iteration} to solve $(X,\yy,S) \in \textit{intF}$ with $XS=\mu I$\;
  }
 \caption{ {\tt Interior\_Point\_Method}\label{alg:inter}
 }
\end{algorithm}

\subsection{Constructive Proof of Theorem \ref{the:2} Using \sdp}

\begin{algorithm}[!htb]
  \SetKwData{Left}{left}\SetKwData{This}{this}\SetKwData{Up}{up}
  \SetKwFunction{Union}{Union}\SetKwFunction{FindCompress}{FindCompress}
  \SetKwInOut{Input}{input}\SetKwInOut{Output}{output}
  \Input{ $\left\{  f_1,\ldots,f_n\right\}, g ,\left\{ h_1,\ldots,h_u \right\}, b $ }
  \Output{ either $\left\{  p_0,\ldots,p_n\right\}$ and $\left\{ q_1,\ldots,q_u \right\}$ such that $1+p_0+p_1f_1+\cdots+p_nf_n+g +q_1 h_1+\cdots+q_uh_u\equiv 0$,  or NULL }
  \SetAlgoLined
  \BlankLine

  Let  $q_{11},q_{12},q_{21},q_{22}, \ldots,q_{u1},q_{u2}\in \mathbf{ SOS}$ with $\deg(q_{i1})\le b$ and  $\deg(q_{i2})\le b $ be undetermined $\mathbf{ SOS}$
  polynomials\;
  Let  $p_1,\ldots,p_n\in \mathbf{SOS} $  with $deg(p_i)\le b$ be undetermined $\mathbf{ SOS }$ polynomials\;
  Let  $f=1+p_0+p_1f_1+\cdots+p_nf_n+g +(q_{11}-q_{12}) h_1+\cdots+(q_{u1}-q_{u2})h_u$\;

  \For{every monomial $m\in f $}{
Let	$\left<Q_m,Q \right>={\tt coeff}(f,m)$\;
  \tcc{ Applying Lemma \ref{lem:cons1} }
	\tcc{ ${\tt coeff}(f,m)$  the coefficient of monomial $m$ in polynomial $f$ }

	\tcc{$Z$ is a monomial vector  that contains all monomials with coefficient 1 and  degree less than or equal to $b/2$ }
	\tcc{$p_0=Z^TQ_0Z,p_1=Z^TQ_1Z,\dots,p_n=Z^TQ_nZ $ }
	\tcc{$q_{i1}=Z^TQ_{i1}Z,q_{i2}= Z^TQ_{i2}Z,i=1,\dots,u$ }
	\tcc{$Q=diag(1,Q_0, Q_1,\dots,Q_n,1,Q_{11},Q_{12},\dots,Q_{u1},Q_{u2} )  $   }

  }
 Applying  {\sdp} software \csdp\ to  solve whether there exists a semi-definite
	 symmetric matrix
 $ Q \ s.t.\  \left<Q_m,Q\right>=0\ $ for every monomial $m\in f$

 \eIf{  the return of \csdp\  is feasible }{
	\tcc{$q_i=q_{i1}-q_{i2} $ }
			\KwRet   $\left\{  p_0,\ldots,p_n\right\},\left\{ q_1,\ldots,q_u \right\}$
		  }{
			\KwRet NULL
		  }
  \caption{\cand\label{alg:sos} }
\end{algorithm}

Given a polynomial $f(\xx)$ of degree no more than $2d$, $f$ can be rewritten as $f=Z^TQZ$ where $Z$ is a vector consists of all monomials of degrees no more than $d$, e.g., $Z=\left[ 1,x_1,x_2,\dots,x_n,x_1x_2,x_2x_3,\dots,x_n^d \right]^T$, and
$ Q= \begin{pmatrix}		
	 a_{1}                 & \frac{ a_{x_1} }{2}       &  \cdots    &     \frac{a_{x_n} }{2}    \\
	 \frac{a_{x_1}}{2}     & a_{x_1^2}                 &  \cdots    &     \frac{a_{x_1 x_n} }{2}  \\
	  \vdots               & \vdots                    &  \ddots    &      \vdots \\
	  \frac{a_{x_n}}{2}    & \frac{ a_{x_1x_n}}{2}     &  \cdots    &      a_{x_n^d}
	   \end{pmatrix}$
is a  symmetric matrix.
	 Note that here $Q$  is not unique in general. Moreover, $f\in \mathcal{C}(\emptyset)$ iff there is a positive semidefinite constant matrix $Q$ such that
$f(\xx)=Z^TQZ.$  The following lemma is an obvious fact on how to use the above notations to express the polynomial multiplication.
 \begin{lemma} \label{lem:cons1}
  For given polynomials $f_1,\dots,f_n, g_1,\dots,g_n$, assume
   $\sum_{i=1}^n f_ig_i=\sum_{i=1}^{s}c_im_i$,  where $c_i\in \RR$ and $m_i$s are monomials.
Suppose $g_i=Z^TQ_{2i}Z$ and $ Q_2=diag(Q_{21},\dots,Q_{2n})$.
 Then there exist  symmetric matrices $Q_{11},\dots,Q_{1s} $ such that $c_i=\left<Q_{1i},Q_{2}\right>$, i.e.,
 $\sum_{i=1}^nf_ig_i=\sum _{i=1}^s \left<Q_{1i},Q_{2}\right>m_i$, in which $Q_{1i}$ can be constructed from the coefficients of  $f_1,\dots,f_n$.
 \end{lemma}

 \begin{example}
 Let $f=a_{20}x_1^2+a_{11}x_1x_2+a_{02}x_2^2$ and $g=b_{00}+b_{10}x_1+b_{01}x_2$. Then,
 $fg$=$\left<Q_{11},Q_2\right> x_1^2 + \left< Q_{12}, Q_2\right> x_1x_2+ \left<Q_{13},Q_2 \right> x_2^2 + \left<Q_{14}, Q_2\right>x_1x_2^2$
 $\left<Q_{15}, Q_2\right> x_1^2x_2 + \left<Q_{16}, Q_2\right> x_2^3 + \left<Q_{17},Q_2\right>x_1^3$
, where
 {\small \[\begin{array}{llll}
   Q_2=   \begin{pmatrix}
   b_{00}                   & \frac{b_{10}}{2} & \frac{b_{01}}{2}  \\
	 \frac{ b_{10}}{2}   & 0                & 0              \\
	 \frac{b_{01} }{2}   & 0                & 0
   \end{pmatrix}, &
 Q_{11}=\begin{pmatrix}
      a_{20}     & 0                & 0              \\
	   0         & 0                & 0              \\
	   0         & 0                & 0
   \end{pmatrix}£¬ &
   Q_{12}=\begin{pmatrix}
      a_{11}     & 0                & 0              \\
	   0         & 0                & 0              \\
	   0         & 0                & 0
   \end{pmatrix}, &
   Q_{13}=\begin{pmatrix}
      a_{02}     & 0                & 0              \\
	   0         & 0                & 0              \\
	   0         & 0                & 0
   \end{pmatrix}, \\
   Q_{14}=\begin{pmatrix}
   0                        &    \frac{a_{02}}{2}  & \frac{a_{11}}{2}  \\
   \frac{a_{02}}{2}         & 0                    & 0              \\
   \frac{a_{11}}{2}         & 0                    & 0
   \end{pmatrix}, &
   Q_{15}=\begin{pmatrix}
   0                        &   \frac{a_{11}}{2}   & \frac{a_{20}}{2}  \\
   \frac{a_{11}}{2}         & 0                    & 0              \\
   \frac{a_{20}}{2}         & 0                    & 0
   \end{pmatrix}, &
    Q_{16}=\begin{pmatrix}
   0                        & 0                    & \frac{a_{02}}{2}  \\
   0                        & 0                    & 0              \\
   \frac{a_{02}}{2}         & 0                    & 0
   \end{pmatrix}£¬ &
Q_{17}=\begin{pmatrix}
  0                        &     \frac{a_{02}}{2}      & 0  \\
  \frac{a_{02}}{2}         & 0                         & 0              \\
   0                       & 0                         & 0
   \end{pmatrix}.
 \end{array}\] }
 \end{example}

Back to Theorem \ref{the:2}. We show how to find $f\in  \mathcal{C}$, $g\in \textit{Mult}$, $h \in \II $ such that $f+g^2+h\equiv 0$ via {\sdp} solving.
First, since $f\in {\mathcal C}$, $f$ can be written as
a sum of the products of some known polynomials and  some  unknown SOSs. Second, $h\in \mathcal{I}(\{h_1,\dots,h_u\})$ is equivalent to $
 h=h_1p_1+\cdots+h_u p_u$, which is further equivalent to $ h= h_1(q_{11}-q_{12})+\dots+ h_u(q_{u1}-q_{u2})$, where $p_i, q_{ij}\in \RR[\xx]$ and $q_{ij} \in \mathbf{SOS}$\footnotemark.
 \footnotetext{For example, let $q_{i1}=(\frac{1}{4}p_i +1)^2, q_{i2}=(\frac{1}{4}p_i -1)^2$.}
 Third, fix an integer $d>0$, let
 $g=(\Pi_{i=1}^tg_i )^{d}$,  and then $f+g^2+h\equiv 0$  can be written as
 $\sum_{i=1}^lf'_i \delta_i$, where $l$ is a constant integer, $f'_i\in
 \RR[\xx]$ are known polynomials and $ \delta_i\in \mathbf{SOS}$ are
 undermined $\mathbf{SOS}$ polynomials.
 Therefore,  Theorem \ref{the:2}  is reduced to fixing a sufficiently
 large integer $d$ and finding   undetermined  $\mathbf{ SOS }$
 polynomials $\delta_i$ occurring in $f,h$  with degrees less than or equal to
 $\deg(g^2)$, which satisfies  $f+g^2+h\equiv 0$. Based on Lemma
 \ref{lem:cons1}, this is a {\sdp} problem of form (\ref{eq:primal}).
 The constraints of the {\sdp} are of the form  $\left<A_j,X\right>=0$,
    where $A_j $ and $ X$ correspond to $Q_{1j}$  and  $Q_2$   in Lemma \ref{lem:cons1}, respectively.
And  $Q_2$ is a block diag  matrix  whose blocks  correspond to the
undetermined  $\mathbf{ SOS }$ polynomials in the above discussion. That is,

\begin{theorem}[\cite{Parrilo00}]\label{the:3}
Consider a system of polynomial equalities and inequalities of the form in Theorem \ref{the:2}. Then the search for bounded degree Positivstellensatz refutations can be done using semidefinite programming. If the degree bound is chosen to be large enough, then the {\sdp}s will be feasible, and the certificates can be obtained from its solution.
\end{theorem}

Algorithm~\ref{alg:sos} is an implementation of Theorem~\ref{the:3} and we will invoke  Algorithm \ref{alg:sos} as a subroutine later. Note that Algorithm~\ref{alg:sos} is a little different from
the original one in \cite{Parrilo03}, as here we require that $f$ has $1$  as a summand for our specific purpose.

\section{Synthesizing Non-linear Interpolants in General Case}\label{sec:sound}
As discussed before, we only need to consider how to
synthesize interpolants for the following two specific SASs
\begin{eqnarray} \label{SAS}
\T_1=\left \{
    \begin{array}{l}
      f_1(\xx)\geq 0,\ldots, f_{s_1}(\xx)\geq 0, \\
      g_1(\xx)\neq 0, \ldots, g_{t_1}(\xx)\neq 0, \\
	  h_1(\xx)=0,\ldots, h_{u_1}(\xx)=0 \\
    \end{array}
   \right.
   & \hspace*{.1cm} &
\T_2=
\left \{
     \begin{array}{l}
      f_{s_1+1}(\xx)\geq 0, \ldots, f_{s}(\xx)\geq 0, \\
      g_{t_1+1}(\xx)\neq 0,  \ldots, g_{t}(\xx)\neq 0, \\
	  h_{u_1+l}(\xx)=0, \ldots, h_{u}(\xx)=0 \\
    \end{array}
   \right.
\end{eqnarray}
where $\T_1$ and  $\T_2$ do not share any real solutions.

 By Theorems \ref{the:2}\&\ref{the:3}, there exist $f\in  \mathcal{C}(\{f_1,\ldots,f_s\})$,
$g\in \textit{Mult}(\{g_1,\ldots,g_t\})$ and $h\in \II(\{h_1,\ldots,h_u\})$ such that $f+g^2+h\equiv 0$,
  where
\begin{eqnarray*}
 g &= & \Pi_{i=1}^tg_i^{2m}, \\
 h &= & q_1h_1+\cdots+q_{u_1}h_{u_1}+\cdots+q_uh_u, \\
 f & = & p_0+p_1f_1+\cdots +p_sf_s+p_{12}f_1f_2+\cdots+p_{1\ldots s}f_1\ldots f_s.
 \end{eqnarray*}
 in which $q_i$ and $p_i$ are in  $\mathbf{SOS}$.


 If $f$ can be represented by three parts: the first part is an $\mathbf{ SOS }$ polynomial that is greater than 0, the second part is from $\mathcal{C}(\{f_1,\ldots,f_{s_1}\})$, and the last part is from
  $\mathcal{C}(\{f_{s_1+1},\ldots,f_{s}\})$, i.e., $f=p_0+\sum_{v \subseteq  \{1,\ldots ,s_1\}} p_{v}(\Pi_{i\in v}f_i)  +\sum_{v\subseteq \{s_1+1,\ldots,s\}} p_{v}(\Pi_{i\in v}f_i)$, where $\forall \xx\!\in \!\RR^n.p_0(\xx) >0$ and $p_v \in \mathbf{SOS}$. Then let
\begin{eqnarray*} \vspace*{-.5cm}
 f_{ \T_1} & = & \sum_{v \subseteq  {1,\ldots ,s_1}} p_{v}\Pi_{i\in v}f_i, ~~~~~~~~~
 h_{\T_1} ~ = ~ q_1h_1+\cdots+q_{u_1}h_{u_1}, \\
 f_{\T_2} & = & \sum_{v \subseteq  {s_1+1,\ldots ,s}} p_{v}\Pi_{i\in v}f_i, ~~~ ~~~~
  h_{\T_2} ~ = ~ h -h_{\T_1},  \\
  q & = & f_{\T_1}+g^2+h_{\T_1}+\frac{q_0}{2} =-(f_{\T_2}+h_{\T_2})-\frac{q_0}{2}.
\end{eqnarray*}
Obviously, we
 have $\forall \xx \!\in \! \T_1. q(\xx)> 0$ and $\forall \xx \! \in \! \T_2. q(\xx)< 0$.
Thus, let $I=q(\xx)>0$. We have $\T_1 \models I$ and $I\wedge \T_2\models \perp$.

Notice that because the requirement on $f$ cannot be guaranteed in general, the above approach is not complete generally.
We will discuss under which condition the requirement can be guaranteed in the next section.
We implement the above method for synthesizing non-linear interpolants in general case by Algorithm~\ref{alg:syn-nonlinear}.

\begin{algorithm}[!htb]
  \SetKwData{Left}{left}\SetKwData{This}{this}\SetKwData{Up}{up}
  \SetKwFunction{Union}{Union}\SetKwFunction{FindCompress}{FindCompress}
  \SetKwInOut{Input}{input}\SetKwInOut{Output}{output}
  \Input{ $\T_1 \mbox{ and } \T_2 \mbox{ of the form }(\ref{SAS}), b $ }
  \Output{ An interpolant $\I$  or NULL }
  \SetAlgoLined
  \BlankLine

$g:=\Pi_{k=1}^t g_k^2$

$g:=g^{ \lfloor\frac{b}{\mathrm{deg}(g)}\rfloor }$

$\left\{ f_{t_1} \right\}:= \left\{  \Pi_{i\in v}f_i  \mbox{ for }  v\subseteq \left\{ 1,\ldots,s_1 \right\}  \right\}$ \;
$\left\{ f_{t_2} \right\}:= \left\{  \Pi_{i\in v}f_i  \mbox{ for }  v\subseteq \left\{ s_1+1,\ldots,s \right\}  \right\}$\;
$V_1=\mathcal{V}(\left\{ f_{t_1} \right\} \cup \left\{	h_1,\dots,h_{u_1} \right\})$\;
\tcc{Get all variables in polynomials system }
$V_2=\mathcal{V}(\left\{ f_{t_2} \right\} \cup \left\{	h_{u_1+1},\dots,h_{u} \right\})$\;
$V=V_1 \cap V_2$\;
$\left(  \left\{ f_{t_1},f_{t_2} \right\}, g, \left\{ h_1,\ldots,h_u \right\} \right)$ = {\tt subvariable}( $\left\{ f_{t_1},f_{t_2} \right\}$, $g$, $\left\{ h_1,\ldots,h_u \right\}$,$V$,$\left\{ h_{1},\dots,h_u \right\}$)\;
\tcc{ Replace  every uncommon variable $v$  by  a polynomial $h$ where $v=h$   as described  in Subsection 3.2}
sdp:=\cand ($\left\{ f_{t_1},f_{t_2} \right\}$, $g$, $\left\{ h_1,\ldots,h_u \right\} ,b$)

		  \eIf{ \emph{sdp} $\equiv$ \emph{NULL} }{
			\KwRet NULL
		  }{
			$\I:= \frac{1}{2}+\sum_{v \subseteq  \{1,\ldots ,s_1\}} p_{v}\Pi_{i\in v}f_i+q_1h_1+\cdots+q_{u_1}h_{u_1}+g>0 $\;
			\KwRet $I$\;
		  }
 \caption{\alg\label{alg:syn-nonlinear}
 }
\end{algorithm}

\begin{example} \label{exm:1}
Consider
\begin{eqnarray*}
\T_1=\left \{ \begin{array}{l}
    x_1^2+x_2^2+x_3^2-2\ge 0, \\
    x_1+x_2+x_3\neq 0, \\
    1.2x_1^2+x_2^2+x_1x_3=0
    \end{array}
  \right.
 & \mbox{ and } &
 \T_2=\left \{  \begin{array}{l}
   -3x_1^2-4x_2^3-10x_3^2+20\ge 0, \\
    2x_1+3x_2-4x_3\neq 0, \\
     x_1^2+x_2^2-x_3-1=0
    \end{array} \right.
  \end{eqnarray*}

 	\begin{Figure}
	 \centering
	  \includegraphics[width=2.6in,height=1.6in]{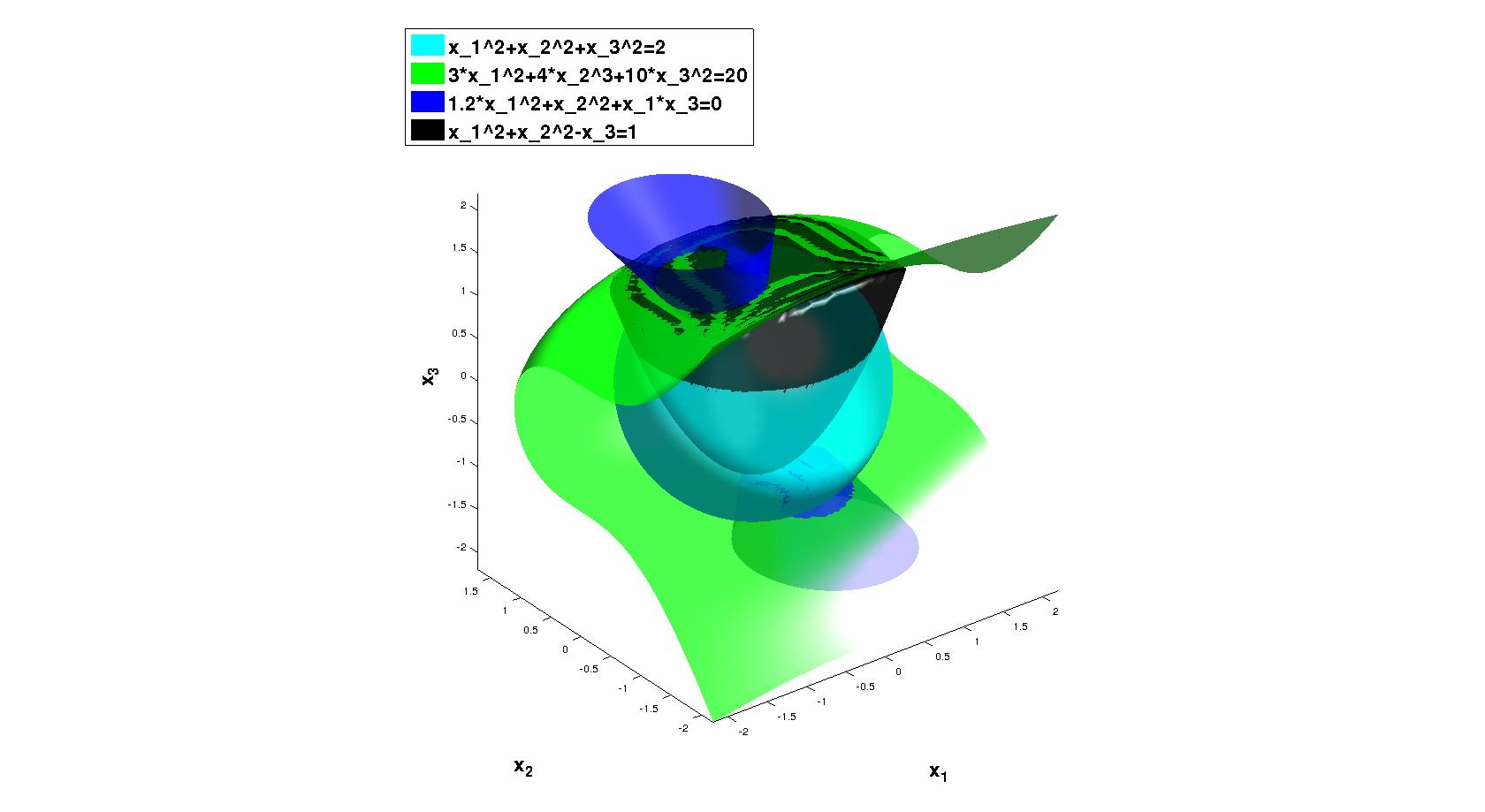}  \hspace*{-1.5cm}
      \includegraphics[width=2.6in,height=1.6in]{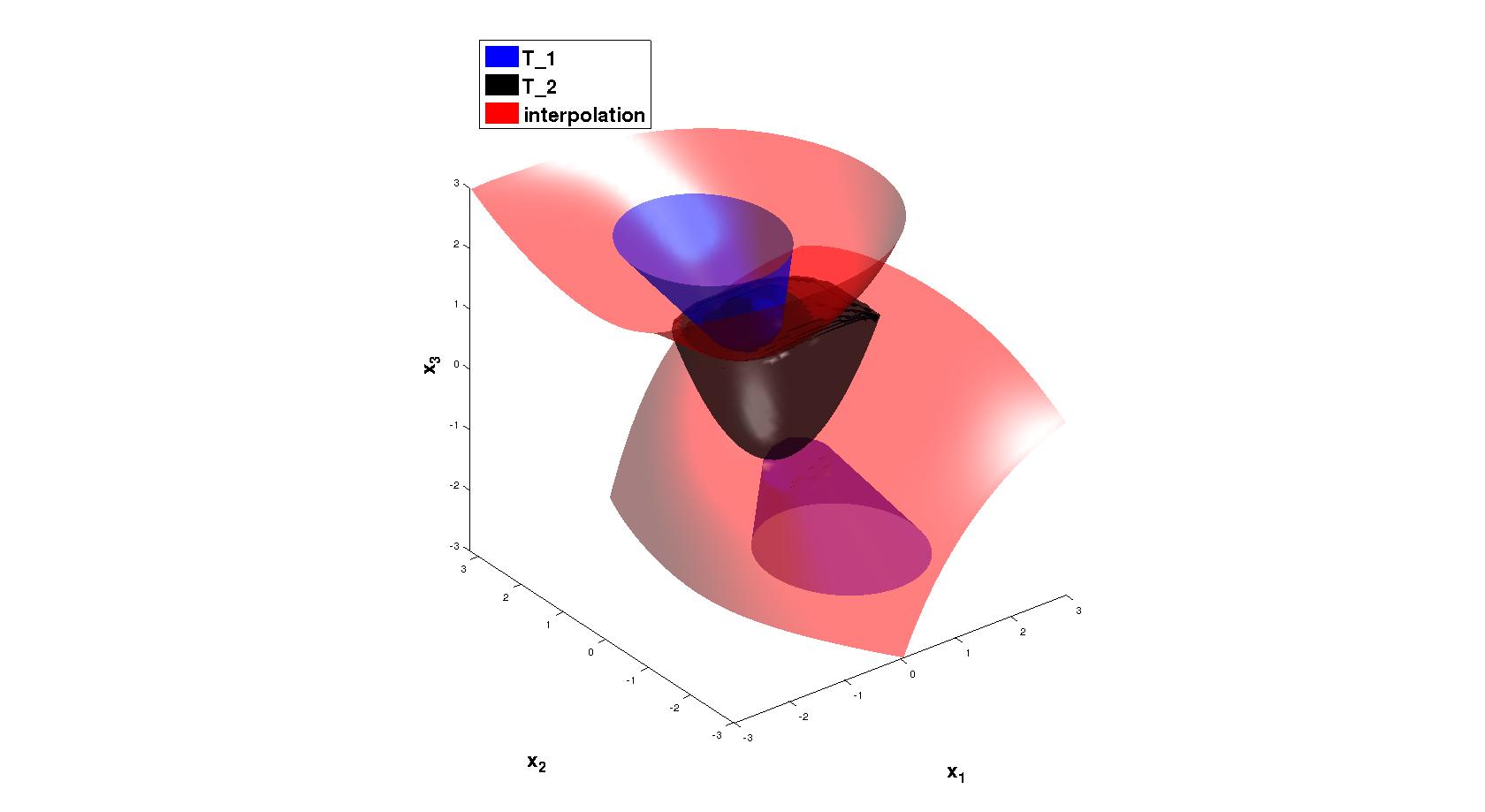}
	  \captionof{figure}{Example \ref{exm:1} \label{fig:1} }
     \end{Figure}

Clearly, $\T_1$ and $\T_2$ do not share any real solutions, see Fig.~\ref{fig:1} (left) \footnotemark.
\footnotetext{For simplicity, we do not draw $ x_1+x_2+x_3\neq 0$, nor $2x_1+3x_2-4x_3\neq 0$
 in the figure.}
 By setting $b=2$, after calling \cand, we obtain an interpolant $\I$ with $30$ monomials
{\small
$-14629.26 +2983.44x_3 +10972.97x_3^2 +297.62x_2 +297.64x_2x_3 +0.02x_2x_3^2 +9625.61x_2^2 -1161.80x_2^2x_3 +0.01x_2^2x_3^2 +811.93x_2^3 +2745.14x_2^4 -10648.11x_1 +3101.42x_1x_3 +8646.17x_1x_3^2 +511.84x_1x_2 -1034.31x_1x_2x_3 +0.02x_1x_2x_3^2 +9233.66x_1x_2^2 +1342.55x_1x_2^2x_3 -138.70x_1x_2^3 +11476.61x_1^2 -3737.70x_1^2x_3 +4071.65x_1^2x_3^2 -2153.00x_1^2x_2 +373.14x_1^2x_2x_3 +7616.18x_1^2x_2^2 +8950.77x_1^3 +1937.92x_1^3x_3 -64.07x_1^3x_2 +4827.25x_1^4$, }
  whose figure is depicted in Fig.~\ref{fig:1} (right).  \qed
\end{example}

\section{A Complete Algorithm Under Archimedean Condition} \label{sec:comp}
Our approach to synthesizing non-linear interpolants presented in Section~\ref{sec:sound} is incomplete generally as it requires that the polynomial $f$ in $\mathcal{C}(\{f_1,\ldots,f_s\})$
  produced by Algorithm~\ref{alg:sos}  can be represented by the sum of three polynomials, one of which
  is positive, the other two polynomials are respectively from $\mathcal{C}(\{f_1,\ldots,f_{s_1}\})$
   and $\mathcal{C}(\{f_{s_1+1},\ldots,f_s\})$.
   In this section, we show, under Archimedean condition, the requirement can be  indeed guaranteed. Thus, our approach will become complete. In particular, we shall argue Archimedean condition
   is a necessary and reasonable restriction in practice.

\subsection{Archimedean Condition}
To the end, we need more  knowledge of real algebraic geometry.
\begin{definition}[quadratic module]
For $g_1,\ldots,g_m\in \RR[\xx]$, the set
\begin{equation}
 \mathcal{M}(g_1,\ldots,g_m) = \{ \delta_0+\sum_{j=1}^m \delta_jg_j\  | \ \delta_0,\delta_j\in \mathcal{C}(\emptyset) \}
  \label{eq:module}
\end{equation}
is called the \emph{quadratic module} generated by $g_1,\ldots,g_m$.
A quadratic module $\mathcal{M}$ is called \emph{proper} if $-1\notin \MM$ (i.e.  $ \MM\neq \RR[\xx]$).
A quadratic module $\mathcal{M}$ is \emph{maximal} if for any $p \in \RR[\xx]\cap \overline{\MM}$,
   $\MM \cup \{p\}$ is not a quadratic module.
\end{definition}

In what follows, we will use $-\MM$ to denote $\{-p \mid p\in \MM\}$ for any given quadratic module
 $\MM$.

The following results are adapted from \cite{laurent} and will be used later, whose proofs can be found in \cite{laurent}.
\begin{lemma}[\cite{BCR98,laurent}]  \label{lem:max}
\begin{description}
\item[1)] If $\MM\subseteq \RR[\xx]$ is a quadratic module, then $I=\MM\cap -\MM$ is an ideal. \label{lem:mod}
\item[2)] If $\MM\subseteq \RR[\xx]$ is a maximal proper quadratic module, then $\MM\cup -\MM=\RR[\xx]$.
 \item[3)]
   $\{ \xx\in \RR^n\ | \ f(\xx)\geq 0\}$ is a compact set\footnote{$S$ is a compact  set  in $\RR^n$ iff $S$ is a  bounded closed set.  } for some  $ f\in  \MM(\{f_1,\ldots,f_s\})$ iff
   \begin{eqnarray}
 & & \forall p \in \RR[\xx], \exists n\in \NN. n \pm p\in  \MM(f_1,\ldots,f_s).
    \label{cond:3}
\end{eqnarray}
 \end{description}
\end{lemma}

\begin{definition}[Archimedean]
  For $g_1,\ldots,g_m\in \RR[\xx]$, the quadratic module $\MM(g_1,\dots,g_m)$ is said to be {\rm Archimedean} if the condition (\ref{cond:3}) holds.
\end{definition}

Let
\begin{eqnarray}\label{eq:t11}
\T_1 =  f_1(\xx)\geq 0,\ldots, f_{s_1}(\xx)\geq 0
  & \mbox{ and }  &
 \T_2=  f_{s_1+1}(\xx)\geq 0, \ldots,  \ f_{s}(\xx)\geq 0
\end{eqnarray}
be two SASs, where $ \left\{ f_i(\xx) \mid i=1,\ldots, s \right\} $ contains constraints on the upper and lower bounds of every variable $x_i$, and  $\T_1$  and $\T_2 $ do not share real solutions.

\begin{remark} \label{rem:1}
   Regarding $\left\{ f_1, \ldots, f_s \right\}$ in  (\ref{eq:t11}), as  every variable is bounded, assume $N-\sum_{i=1}^bx_i^2\in \{f_1,\ldots,f_s\}$ for a const $N$, then  $\MM(f_1,\ldots,f_s)$ is {\it Archimedean}.
  \end{remark}

\begin{lemma}\cite{BCR98,laurent}
  Let $\MM\subseteq \RR[\xx]$ be a maximal proper quadratic module which is {\rm Archimedean}, $I=\MM\cap-\MM$, and $f\in \RR[\xx]$,  then there exists $a \in \RR$ such that $f-a\in I$.
  \label{lem:exi}
\end{lemma}

\begin{lemma}
  If $I$ is an ideal and there exists $\bfa=(a_1,\ldots,a_n)\in \RR^n$ such that $x_i-a_i\in I$ for $i=1,\ldots,n$, then for any $f\in \RR[\xx], f-f(\bfa)\in I$.
  \label{lem:in}
\end{lemma}

\begin{proof}
  Because $x_i-a_i\in I$ for $i=1,\ldots,n$,  $\left<x_1-a_1,\ldots,x_n-a_n\right>\subseteq I$.
   For any $f\in \RR[\xx]$, $\left<x_1-a_1,\ldots,x_n-a_n\right>$ is
  a radical ideal\footnote{Ideal $I$ is a radical ideal if $I=\sqrt{I}=\{f|f^k\in I\ \mbox{for some integer}\ k\ge 0\}$.} and $(f-f(\bfa))(\bfa)=0$, so $f-f(\bfa)\in \left<x_1-a_1,\ldots,x_n-a_n\right>\subseteq I$.
   \qed
\end{proof}

\begin{theorem} \label{thm:4}
Suppose  $\left\{ f_1(\xx),\ldots, f_s(\xx) \right\}$  is given in (\ref{eq:t11}).
 If $\bigwedge_{i=1}^s(f_i\ge0)$ is unsatisfiable, then $-1\in  \MM(f_1,\ldots,f_s)$.
  \label{the:mine}
\end{theorem}

\begin{proof}
  By Remark \ref{rem:1}, $\MM\left( f_1,\ldots,f_s \right)$ is {\it Archimedean}.
 Thus, we only need to prove that the quadratic module $\MM(f_1,\ldots,f_s)$ is not proper.

  Assume  $\MM(f_1,\ldots,f_s)$ is proper. By Zorn's lemma, we can extend $\MM(f_1,\ldots,f_s)$ to a maximal proper quadratic module $\MM\supseteq \MM(f_1,\ldots,f_s)$. As $\MM(f_1,\ldots,f_s)$ is  Archimedean, $\MM$ is also Archimedean.
   By Lemma \ref{lem:exi}, there exists $\bfa=(a_1,\ldots,a_n)\in \RR^n$ such that $x_i-a_i\in I=\MM\cap -\MM$ for all $i\in \{1,\ldots,n\}$. From Lemma \ref{lem:in}, $f-f(\bfa)\in I$ for any $f\in \RR[\xx]$. In particular, for $f=f_j$, we have $f_j(\bfa)=f_j-(f_j-f_j(\bfa)) \in \MM$ since $f_j\in \MM(f_1,\ldots,f_s) \subseteq \MM$ and $-(f_j-f_j(\bfa))\in \MM$, which implies $f_j(\bfa)\geq 0$, for $j=1,\dots,s$. This contradicts to
   the unsatisfiability of  $\bigwedge_{i=1}^s(f_i\ge0)$. \qed
\end{proof}

 By Theorem \ref{the:mine} we have $-1 \in  \MM(f_1,\ldots,f_s )$.  So, there exist
 $\sigma_0,\ldots,\sigma_{s} \in \CC(\emptyset)$  such that
$ -1 = \sigma_0+\sigma_1f_1+\cdots +\sigma_{s_1}f_{s_1}+
\sigma_{s_1+1}f_{s_1+1}+\cdots+f_s\sigma_s.$
It follows
 \begin{eqnarray} \label{eq-ar1}
 -(\frac{1}{2}+\sigma_{s_1+1}f_{s_1+1}+\cdots+\sigma_sf_s) & = & \frac{1}{2}+\sigma_0+\sigma_1 f_1+\cdots+\sigma_{s_1}f_{s_1}.
  \end{eqnarray}
  Let $q(\xx)=\frac{1}{2}+\sigma_0+\sigma_1f_1+\cdots + \sigma_{s_1}f_{s_1}$, we have
$\forall \xx \in \T_1. q(\xx)> 0$ and $\forall \xx \in \T_2. q(\xx)< 0$. Thus,
let $I= q(\xx)>0$.
According to Definition \ref{def:int}, $I$ is an interpolant of $\T_1$ and $\T_2$. So, under Archimedean condition, we can revise Algorithm~\ref{alg:syn-nonlinear} as Algorithm~\ref{alg:ualg}.

\begin{algorithm}[!htb]
  \SetKwData{Left}{left}\SetKwData{This}{this}\SetKwData{Up}{up}
  \SetKwFunction{Union}{Union}\SetKwFunction{FindCompress}{FindCompress}
  \SetKwInOut{Input}{input}\SetKwInOut{Output}{output}
  \Input{ $\T_1$ and $T_2$ as in (\ref{eq:t11}), $\left\{ h_{u_1+1},\dots,h_u \right\}$ }
  \tcc{ $h_{u_1+1},\dots,h_u$ are the equality occur in $T_2$}
  \Output{ $\I$   }
  \SetAlgoLined
  \BlankLine
  \SetKw{KwGoTo}{go to}

b=2\;
$V_1=\mathcal{V}( \left\{	f_1,\dots,f_{s_1} \right\})$\;
\tcc{Get all variables of $\T_1$}
$V_2=\mathcal{V}( \left\{	f_{s_1+1},\dots,f_{u} \right\})$\;
$V=V_1 \cap V_2$\;

$\left\{	f_{s_1+1},\dots,f_{u} \right\})$={\tt subvariable}($\left\{	f_{s_1+1},\dots,f_{u} \right\})$,$V$, $\left\{ h_{u_1},\dots,h_u \right\}$)\;
\tcc{Replaceing  every uncommon variable $v$   by polynomial $h$ where $v\ge h$ and $v\le h$ as described in Section 3.2}
\While{ true}{\label{outer_loop}

  sdp=\cand ($\left\{ f_1,\ldots,f_s \right\}$,0,$\left\{  \right\}$,$b$)\;

 \eIf{ {\rm sdp} $\neq$ {\rm NULL}  }{
			  $\I=\left\{ \frac{1}{2}+\sum_{i=1}^{s_1}p_if_i>0 \right\}$\;
			  $I'$={\tt subvariable}($I$,$V$, $\left\{ h_{u_1+1},\dots,h_u \right\}$)\;
			  \KwRet $I'$\;
		  }
		  {
			b=b+2\;
		  }
		}
		
  \caption{\ualg \label{alg:ualg}}
\end{algorithm}

\begin{example} \label{exm:2}
 Let $\Psi= \bigwedge_{i=1}^3 x_i\ge-2 \wedge -x_i\ge -2$, $f_1=-x_1^2-4x_2^2-x_3^2+2$, $f_2= x_1^2-x_2^2-x_1x_3-1$,
  $f_3=-x_1^2-4x_2^2-x_3^2+3x_1x_2+0.2$, $f_4=-x_1^2+x_2^2+x_1x_3+1 $. Consider
   $\T_1=  \Psi \wedge f_1\ge0 \wedge f_2\ge0$ and
   $\T_2= \Psi \wedge f_3\ge 0 \wedge f_4\ge 0$.  Obviously,
    $\T_1 \wedge T_2$ is unsatisfiable, see Fig.~\ref{fig:3} (left).

    \begin{Figure}
	 \centering
	 \includegraphics[width=2.6in,height=1.6in]{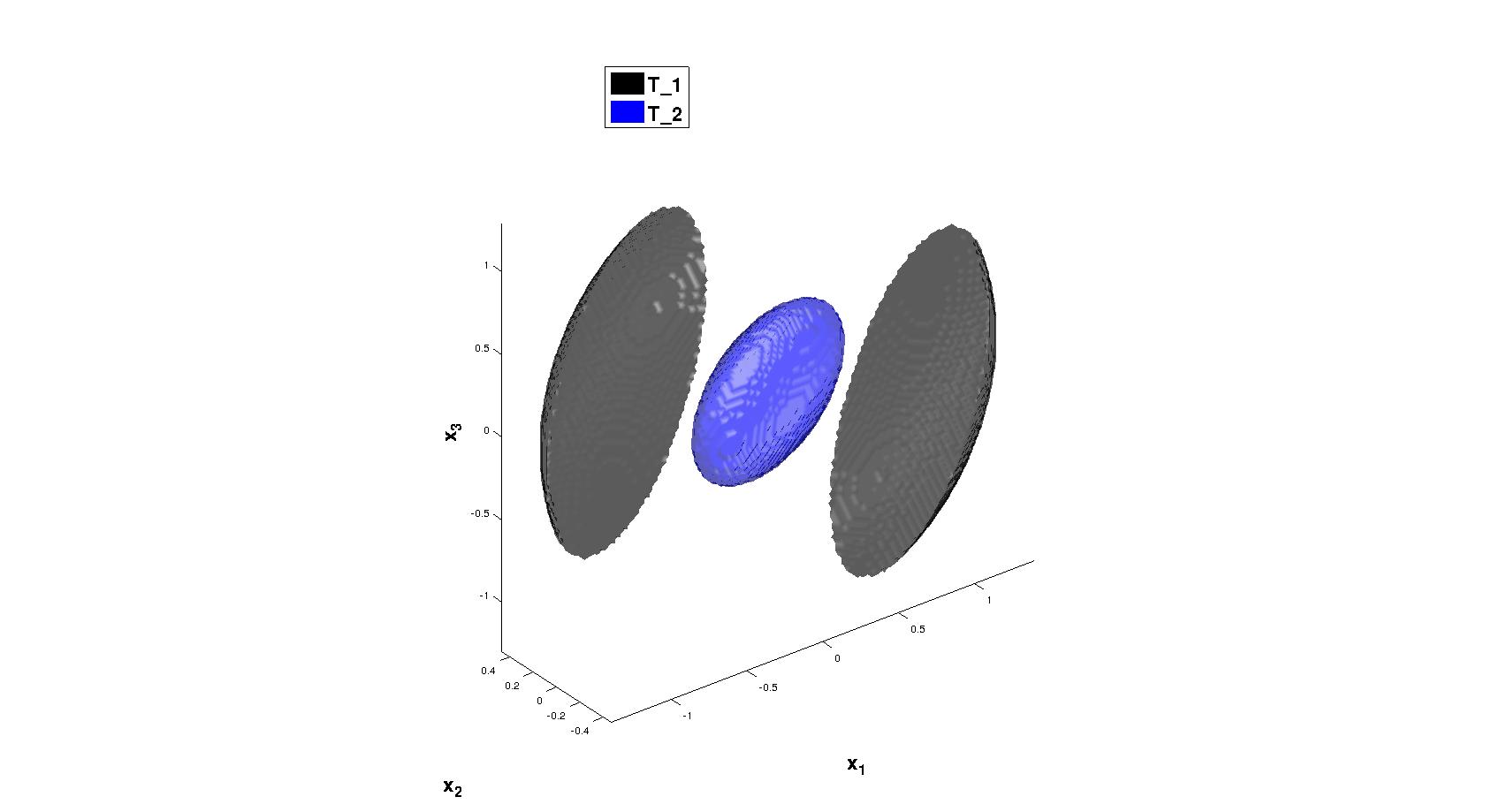} \hspace*{-1.5cm}
      \includegraphics[width=2.6in,height=1.6in]{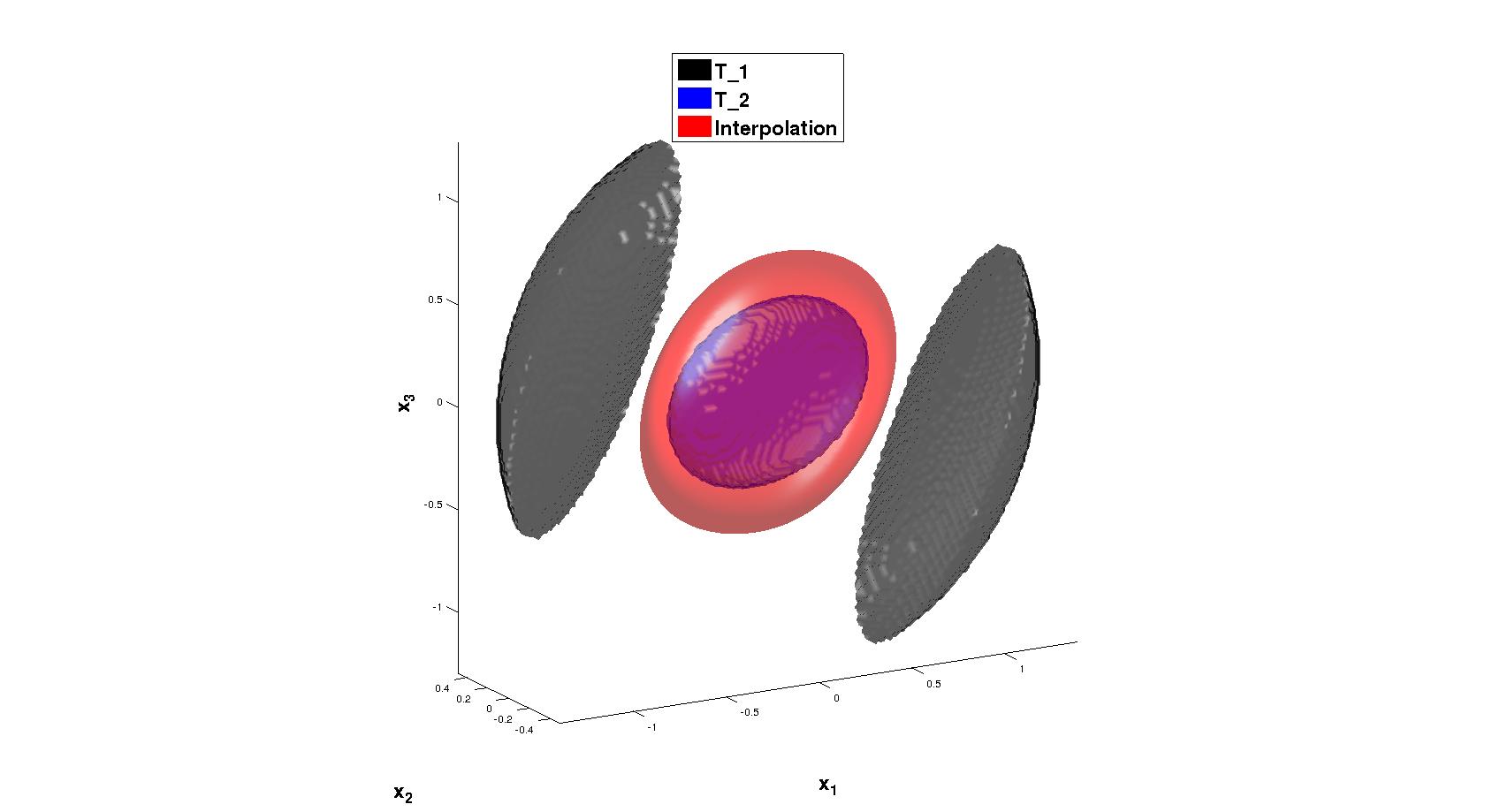}
	  \captionof{figure}{Example \ref{exm:2} \label{fig:3} }
	 \end{Figure}

By applying \ualg, we can get an interpolant as
{\small $ -33.7255x_1^4 + 61.1309x_1^3x_2 + 4.6818x_1^3x_3 - 57.927x_1^2x_2^2
  + 13.4887x_1^2x_2x_3 - 48.9983x_1^2x_3^2 - 8.144x_1^2 - 48.1049x_1x_2^3
  - 6.7143x_1x_2^2x_3 + 29.8951x_1x_2x_3^2 + 61.5932x_1x_2 + 0.051659x_1
  x_3^3 - 0.88593x_1x_3 - 34.7211x_2^4 - 7.8128x_2^3x_3 - 71.9085x_2^2x_3^2
  - 60.5361x_2^2 - 1.6845x_2x_3^3 - 0.5856x_2x_3 - 15.2929x_3^4 - 9.7563x_3^2 + 6.7326$,}
  which is  depicted in Fig \ref{fig:3} (right).
 In this example, the final value of $b$ is $2$. \qed
\end{example}

\subsection{Discussions}

\paragraph{\bf 1. Reasonability of Archimedean condition:}

Considering only bounded numbers can be represented in computer, so it is reasonable to constraint each variable with upper and lower bounds in practice. Not allowing strict inequalities indeed reduce the expressiveness from a theoretical point of view. However, as only numbers with finite precision can be represented in computer,
we always can relax a strict inequality to an equivalent non-strict inequality in practice. In a word, we believe
 \emph{Archimedean condition} is reasonable in practice.

\paragraph{\bf 2. Necessity of Archimedean condition: } In Theorem~\ref{thm:4}, \emph{Archimedean condition}
is necessary. For example,
let $\T_1=\{x_1\ge 0, x_2\ge 0\}$ and $\T_2= \{-x_1x_2 -1 \geq 0\}$.
Obviously, $\T_1\wedge\T_2=\emptyset$ is not \emph{Archimedean} and unsatisfiable, but

\begin{theorem}
 $-1 \not \in M(x_1,x_2,-x_1x_2-1)$.
 \end{theorem}
\begin{proof}
   Suppose these exist $\delta_0,\delta_1,\delta_2,\delta_3\in\CC(\emptyset)$ such that $h=\delta_0+x_1\delta_1+x_2\delta_2-(x_1x_2+1)\delta_3=-1$.
    Let $c_0x_1^{2a_0}x_2^{2b_0}$, $ c_1x_1^{2a_1+1}x_2^{2b_1}$, $ c_2x_1^{2a_2}x_2^{2b_2+1}$, and  $c_3x_1^{2a_3+1}x_2^{2b_3+1}$ be the leading terms of
$\delta_0$, $x_1\delta_1$, $x_2\delta_2$ and $(x_1x_2+1)\delta_3$, respectively,  according to
  the total degree order of monomials, where
 $c_i\ge 0$
 and $a_i,b_i\in\NN$.
  Obviously, the four terms are pairwise different. So, the leading term of $h$ must be one of the four terms if they are not zero. This, together with $h=-1$, imply that $c_1=c_2=c_3=0$ and thus $\delta_1=\delta_2=\delta_3=0$. Therefore, $\delta_0=-1$, a contradiction.
\qed
\end{proof}

\section{Correctness and Complexity Analysis} \label{sec:cor}
The correctness of the algorithm \alg is obvious according to Theorem~\ref{the:3} and the discussion of Section 4. Its complexity just corresponds to one iteration of  the algorithm \ualg.
The correctness of  the algorithm {\ualg} is guaranteed by Theorem~\ref{the:3} and Theorem~\ref{thm:4}.
  The cost of  each iteration of \ualg\  depends on  the number of the variables $n$, the number of polynomial constraints $u$, and the current value of $b_f$. The size of $X$ in (\ref{eq:primal})
is $u\binom{n+b_f/2}{n}$ and the $m$ in (\ref{eq:primal}) is $\binom{n+b_f}{n}$. So, the complexity of applying interior method to solve the  {\sdp} is polynomial in  $ u\binom{n+b_f/2}{n} \binom{n+b_f}{n}$.
Hence, the cost of each iteration of \ualg\ is  $ u  \binom{n+b_f/2}{n}\binom{n+b_f}{n}$. Therefore, the total cost of \ualg\ is $b_f u \binom{n+b_f/2}{n}\binom{n+b_f}{n}$. For a given
problem, $n,u$ are fixed, so  the complexity of the algorithm becomes
polynomial in  $b_f$.
The complexity of Algorithm  \alg\ is the same as above discussions, except that the number of polynomial constraints is about $2^{s_1}+2^{s-s_1}$.

 As indicated in \cite{Parrilo03}, there are upper bounds on $b_f$, which are at least triply exponential. So our approach can enumerate all possible instances, but can not be done in polynomial time.

\section{Implementation and Experimental Results}\label{sec:exp}
We have implemented a prototypical tool of the algorithms described in this paper, called \aisat, which contains
  6000 lines of {\tt C++} codes.  \aisat\ calls {\tt Singular} \cite{singular} to deal with polynomial input and
\csdp\ to solve {\sdp}s. In \aisat, we design a specific algorithm to
 transform polynomial constraints to matrices constraints,  which indeed improves the efficiency of our tool
  very much, indicated by the comparison with \sostool\cite{PPP02} (see the table below). 
As a future work, we plan to implement a new {\sdp} solver with more  stability   and  convergence efficiency
 on solving {\sdp}s.

In the following, we report some experimental results by applying \aisat\ to some
benchmarks.

The first example is from \cite{bench1}, see the source code in Code~\ref{ex1}. We show its correctness by
applying \aisat \, to the following two possible executions.
 \begin{itemize}
  \item Subproblem $1$: Suppose  there is an execution starting from a state satisfying the assertion at line
   $13$ (obviously, the initial state satisfies the assertion), after $\rightarrow 6\rightarrow 7 \rightarrow 8 \rightarrow 9 \rightarrow 11 \rightarrow 12 \rightarrow 13$, ending at a state that does not satisfy
   the assertion.

  Then the interpolant synthesized by our approach is
   $716.77 +1326.74(ya) +1.33(ya)^2 +433.90(ya)^3 +668.16(xa) -155.86(xa)(ya) +317.29(xa)(ya)^2 +222.00(xa)^2 +592.39(xa)^2(ya) +271.11(xa)^3$, which guarantees that this execution is infeasible.
\item  Subproblem $2:$   Assume there is an execution starting from a state satisfying the assertion at line
   $13$, after $\rightarrow 6\rightarrow 7 \rightarrow 8 \rightarrow 10 \rightarrow 11 \rightarrow 12 \rightarrow 13$, ending at a state that does not satisfy
   the assertion.

 The interpolant generated by our approach is
{\small $716.95 +1330.91(ya) +67.78(ya)^2 +551.51(ya)^3 +660.66(xa) -255.52(xa)(ya) +199.84(xa)(ya)^2 +155.63(xa)^2 +386.87$ $(xa)^2(ya) +212.41(xa)^3$}, which guarantees this execution is infeasible either.
 \end{itemize}

\begin{tabular}[!htb]{ p{0.45\textwidth} p{0.55\textwidth} }
  \begin{lstlisting} [label=ex1,caption=ex1]
  int main () {
  int x,y;
  int xa := 0;
  int ya := 0;
  while (nondet()) {
	x := xa + 2*ya;
	y := -2*xa + ya;
	x++;
	if (nondet()) y= y+x;
	else y := y-x;
	xa := x - 2*y;
	ya := 2*x + y;}
  assert (xa + 2*ya >= 0);
  return 0;
}
  \end{lstlisting}

  &
\begin{lstlisting}[label=car,caption= An accelerating car]
vc:=0;
 /* the initial veclocity */
fr:=1000;
  /* the initial force */
ac:=0.0005*fr;
 /* the initial acceleration */
while ( 1 ) {
	fa:=0.5418*vc*vc;
      /* the force control */
	fr:=1000-fa;
	ac:=0.0005*fr;
	vc:=vc+ac;
	assert(vc<49.61);
       /* the safety velocity */
}
  \end{lstlisting}
  \tabularnewline
\end{tabular}

The second example {\tt accelerate} (see Code~\ref{car}) is from \cite{Kupferschmid11}.
 Taking
the air resistance into account, the relation between the car's velocity
and the physical drag contains quadratic functions. Due to air resistance
the velocity of the car cannot be beyond $49.61m/s$, which is a safety
property.
  Assume that there is an execution $ (vc<49.61)\rightarrow 8 \rightarrow 10 \rightarrow 11 \rightarrow 12 \rightarrow 13 (vc\geq 49.61) $.
 By applying  Applying \aisat, we can obtain an interpolant
$ -1.3983vc + 69.358>0$, which guarantees $vc<49.61$. So, {\tt accelerate} is correct.
we can synthesize an interpolant, which guarantees the safety property.

The last example  {\tt logistic} is also from \cite{Kupferschmid11}.
Mathematically, the logistic loop is written as $x_{n+1}=rx_n(1-x_n)$, where
$0\leq x_n \leq 1$.  When $r=3.2$, the logistic loop oscillates
between two values. The verification obligation is to guarantee that it is within the safe region  $(0.79\le x \wedge x\le 0.81 ) \vee ( 0.49 \le x \wedge x \le 0.51)$. By applying \aisat \, to
the following four possible executions, the correctness is obtained.
  \begin{itemize}
  \item Subproblem $1$:  $\{x\geq 0.79 \wedge x\leq 0.81\}$ {\tt logistic} $\{x>0.51\}$
    is invalidated by the synthesized interpolant
$108.92 -214.56x>0 $.

  \item Subproblem $2$: $\{x\geq 0.79\wedge x\leq 0.81\}$ {\tt logistic} $\{x< 0.49\}$
  is outlawed by  the synthesized interpolant
$ -349.86 +712.97x>0$.

\item Subproblem $3$: $\{x\geq 0.49 \wedge x\leq 0.51 \}$ {\tt logistic} $\{x> 0.81\}$
  is excluded by the generated interpolant
$ 177.21 -219.40x>0$.
  \item Subproblem $4$: $\{x\geq 0.49 \wedge x\leq 0.51\}$  {\tt logistic} $\{x< 0.79\}$ is
   denied by the generated  interpolant
$-244.85 +309.31x >0$.
 \end{itemize}

Some experimental results of applying \aisat\ to the above three examples on a desktop (64-bit Intel(R) Core(TM) i5 CPU 650 @ 3.20GHz, 4GB RAM memory and Ubuntu 12.04 GNU/Linux) are listed in the table below.
Meanwhile, as a comparison, we apply the SOSTOOLS to the three examples with the same computer.

\begin{center}
  \begin{tabular}{ | c |  c | c| c |}
	\hline
	Benchmark & \#Subporblems   & \aisat\ (milliseconds) & \sostool\ (milliseconds)  \\ \hline
	{\tt ex1} & 2 &   60 & 3229 \\ \hline
	{\tt accelerate }& 1 &  940 & 879 \\ \hline
	{\tt logistic} & 4 &  20 & 761 \\
	\hline
  \end{tabular}
\end{center}

\section{Related work} \label{sec:rel}
In Introduction, we have introduced many work related to interpolant generation and its application to
program verification. In this section, we will mention some existing work on program verification, so that
we give a comparison between our approach and them.

Work on program verification can date back to the late sixties (or
early seventies) of the 20th century when  the so-called
\emph{Floyd-Hoare-Naur's inductive assertion
method}\cite{floyd67,hoare69,naur66} was invented, which was thought
as the dominant approach on automatic program verification. The
method is based on Hoare Logic\cite{hoare69}, by using \emph{pre-}
and \emph{post- conditions}, \emph{loop invariants} and
\emph{termination analysis} through ranking functions, etc.
Therefore,
 the discovery of loop invariants and ranking functions plays a
central role in proving the correctness of programs and  is also
thought of as the most challenging part of the approach.

Since then, there have been lots of attempts to handle invariant
generation of programs, e.g. \cite{wegbreit74,gw75,km76,karr76}, but
only with a limited success.
 Recently, due to the advance of computer algebra,
several methods based on symbolic computation have been applied
successfully to invariant generation, for example the techniques
based on abstract interpretation \cite{ch78,bjt99,rck04a,cousot05},
quantifier elimination \cite{css03,kapur04,CXYZ07} and polynomial algebra
\cite{ms04,rck04b,rck07,ssm04a}.

The basic idea behind the approaches based on abstract
interpretation is to perform an approximate symbolic execution of a
program until an assertion is reached that remain unchanged by
further executions of the program. However, in order to guarantee
termination, the method introduces imprecision by use of an
extrapolation operator called \emph{widening/narrowing}. This
operator often causes the technique to produce weak invariants.
Moreover, proposing widening/narrowing operators with certain
concerns of completeness is not easy and becomes a key challenge for
abstract interpretation based techniques \cite{ch78,bjt99}.

In contrast, approaches by exploiting the theory of polynomial
algebra to discover invariants of polynomial programs were proposed
in \cite{ms04,rck04b,rck07,ssm04a}. In \cite{ms04}, Mueller-Olm
and Seidl applied the technique of linear algebra to generate
polynomial equations of bounded degree as invariants of programs
with affine assignments. In \cite{rck04b,rck07}, Rodrigez-Carbonell
and Kapur
 first proved that the set of polynomials serving as
 loop invariants has the algebraic structure of ideal,
 then proposed an invariant generation algorithm by using
 fixpoint computation,
 and finally implemented the
 algorithm by the Gr\"{o}bner bases and the elimination theory.
  The
approach is theoretically sound and complete in the sense that if
there is an invariant of the loop that can be expressed as a
 conjunction of polynomial equations, applying the approach
   can indeed generate it.
 While in \cite{ssm04a}, the authors presented a similar approach to finding invariants represented by
 a polynomial equation  whose form is  priori determined (called templates) by
using an extended Gr\"{o}bner basis algorithm. The complexity of the above approaches are double exponential as
Gr\"{o}bner base technique is adopted.

Compared with  polynomial algebra based approaches that can only
generate
 invariants represented as polynomial equations,
 Col\'{o}n et al in  \cite{css03} proposed an approach to generate linear
 inequalities as invariants for linear programs,
    based on
\emph{Farkas' Lemma} and nonlinear constraint solving. The complexity depends on the complexity of
linear programming, which is in general is polynomial
in the number of variables.

  In addition, Kapur in \cite{kapur04} proposed a very general approach for
automatic generation of more expressive invariants by exploiting the
technique of quantifier elimination, and applied the approach to
Presburger Arithmetic and quantifier-free theory of conjunctively
closed polynomial equations. Theoretically speaking, the approach
can also be applied to the theory of real closed fields, but Kapur
also  pointed out in \cite{kapur04} that this  is impractical in
reality because of the high complexity of quantifier elimination,
which is doubly exponential \cite{dh88} in the number of quantifiers.
While in \cite{CXYZ07}, we improved Kapur's approach by using
 the theory of real root classification of SASs \cite{xy02}, with the complexity
 singly exponential in the number of variables and doubly exponential in the number of
 parameters.

Comparing with the approaches based on polynomial algebra, or Farkas' Lemma or Gr\"{o}bnes, our approach is more powerful, also more efficient except for Farkas' Lemma based approach. Comparing with
quantifier elimination based approach \cite{kapur04,CXYZ07}, our approach is much more efficient, even according to
the complexity analysis of quantifier elimination given in \cite{Brown}, which is doubly exponential
in the number of the quantifier alternation, and becomes singly exponential in the number of variables
and constraints in our setting\footnotemark. \footnotetext{Hoon Hong in \cite{HH91}
 pointed out the existing algorithms for the existential real theory  which are
 singly exponential in the number of variables  is far from realization, even worse than
  the general algorithm for quantifier elimination.}

\section{Conclusion} \label{sec:con}
The main contributions of the paper include:
\begin{itemize}
  \item  We give a sound but not inomplete algorithm \alg for the generation of interpolants for non-linear arithmetic in general.
  \item  If the two systems satisfy {\it Archimedean condition}, we provide a more practical algorithm \ualg, which is not only sound but also complete, for generating Craig interpolants.
    \item We implement the above algorithms as a protypical tool \aisat,
      and demonstrate our approach by applying the tool to some benchmarks.
\end{itemize}

In the future, we will focus on how to combine non-linear arithmetic with other well-established
 decidable first order theories. In particular, we believe that we can use the method of \cite{sofronie06,Kupferschmid11} to extend our algorithm to uninterpreted functions. To investigate errors caused by numerical computation in {\sdp} is quite interesting.  In addition,
 to investigate the possibility to apply our results to the verification of hybrid systems is very significant.

\oomit{\paragraph{\bf Acknowledgements}

The work is partly supported by the National Natural Science Foundation of China (Grant No.11271034, No.11290141, No.)
and the project SYSKF1207 from SKLCS, IOS, the Chinese Academy of Sciences. }

\bibliographystyle{splncs}
\bibliography{semdefinter}

\end{document}